\pdfoutput=1

\documentclass[a4paper,USenglish,cleverref,autoref]{lipics-v2019}
\usepackage[utf8]{inputenc}
\usepackage{amsmath}
\usepackage{amsfonts}
\usepackage{amssymb}
\usepackage{amsthm}
\usepackage{nicefrac}

\nolinenumbers %
\hideLIPIcs  %

\usepackage{graphicx}
\usepackage{subcaption}
\graphicspath{{img/}}

\usepackage{tikz}
\usetikzlibrary{positioning}
\usetikzlibrary{backgrounds}

\usepackage{mathtools}

\DeclareMathOperator{\dist}{dist}

\usepackage{accents}
\newcommand{\parent}[1]{\accentset{\leftharpoonup}{#1}}
\newcommand{\child}[1]{\accentset{\rightharpoonup}{#1}}

\newcommand{\problem}[1]{\textsc{#1}}

\theoremstyle{plain}
\newtheorem{observation}[theorem]{Observation}

\title{Hierarchy of Transportation Network Parameters and Hardness Results}
\author{Johannes Blum}{University of Konstanz, Germany}{johannes.blum@uni-konstanz.de}{https://orcid.org/0000-0003-1102-3649}{}
\authorrunning{J. Blum}%

\Copyright{Johannes Blum}%



\category{}%

\supplement{}%

\EventEditors{Bart M. P. Jansen and Jan Arne Telle}
\EventNoEds{2}
\EventLongTitle{14th International Symposium on Parameterized and Exact Computation (IPEC 2019)}
\EventShortTitle{IPEC 2019}
\EventAcronym{IPEC}
\EventYear{2019}
\EventDate{September 11--13, 2019}
\EventLocation{Munich, Germany}
\EventLogo{}
\SeriesVolume{148}
\ArticleNo{4}

\begin{document}

\maketitle

\begin{abstract}
The graph parameters highway dimension and skeleton dimension were introduced to capture the properties of transportation networks. 
As many important optimization problems like \problem{Travelling Salesperson}, \problem{Steiner Tree} or \problem{$k$-Center} arise in such networks, it is worthwhile to study them on graphs of bounded highway or skeleton dimension.

We investigate the relationships between mentioned parameters and how they are related to other important graph parameters that have been applied successfully to various optimization problems. We show that the skeleton dimension is incomparable to any of the parameters distance to linear forest, bandwidth, treewidth and highway dimension and hence, it is worthwhile to study mentioned problems also on graphs of bounded skeleton dimension. Moreover, we prove that the skeleton dimension is upper bounded by the max leaf number and that for any graph on at least three vertices there are edge weights such that both parameters are equal.

Then we show that computing the highway dimension according to most recent definition is NP-hard, which answers an open question stated by Feldmann et al.~\cite{Fel15b}. Finally we prove that on graphs $G=(V,E)$ of skeleton dimension $\mathcal{O}(\log^2 \vert V \vert)$ it is NP-hard to approximate the \problem{$k$-Center} problem within a factor less than $2$.
\end{abstract}

\section{Introduction}
Many important optimization problems arise in the context of road or flight networks, e.g.\ \problem{Travelling Salesperson} or \problem{Steiner Tree}, and have applications in domains like route planning or logistics. Therefore, several approaches have been developed that try to exploit the special structure of such transportation networks. Examples are the graph parameters highway dimension and skeleton dimension. Intuitively, a graph has low highway dimension $hd$ or skeleton dimension $\kappa$, if there is only a limited number of options to leave a certain region of the network on a shortest path.
Both parameters were originally used in the analysis of shortest path algorithms and it was shown that if $hd$ or $\kappa$ are small, there are preprocessing-based techniques to compute shortest paths significantly faster than the algorithm of Dijkstra \cite{abr10, abr11, abr16, kos16}.

The highway dimension was also investigated in the context of NP-hard optimization problems, such as \problem{Travelling Salesperson} (\problem{TSP}), \problem{Steiner Tree} and \problem{Facility Location}~\cite{Fel15b}, \problem{$k$-Center}~\cite{DBLP:journals/algorithmica/Feldmann19, Fel18, Becker2018} or \problem{$k$-Median} and \problem{Bounded-Capacity Vehicle Routing}~\cite{Becker2018}. It was shown that in many cases, graphs of low highway dimensions allow better algorithms than general graphs. To our knowledge, the skeleton dimension has exclusively been studied in the context of shortest path algorithms so far. However, it was shown that real-world road networks exhibit a skeleton dimension that is clearly smaller than the highway dimension~\cite{blu18}. Moreover, in contrast to the highway dimension, it can be computed in polynomial time. Hence it is natural to study the aforementioned problems on networks of low skeleton dimension.

Further graph classes that have been used to model transportation networks are for instance planar graphs and graphs of low treedwidth or doubling dimension.
Moreover, many important optimization problems have been studied extensively for classic graph parameters like treewidth or pathwidth~\cite{DBLP:conf/icalp/Bodlaender88,DBLP:journals/bit/Arnborg85}. Still, there are only partial results on how the highway dimension $hd$ and skeleton dimension $\kappa$ are related to these parameters. This is the starting point of the present paper. A better understanding of the relationships between $hd$, $\kappa$ and different well-studied graph parameters will allow a deeper insight in the structure of transportation networks and might enable further algorithms custom-tailored for such networks.

\subsection{Related Work}
We now briefly sum up some algorithmic results in the context of optimization problems in transportation networks. Arora~\cite{Arora1998} developed a general framework that enables PTASs for several geometric problems where the network is embedded in the Euclidean plane. Building upon the work of Arora, Talwar~\cite{Talwar2004} developed QPTASs for \problem{TSP}, \problem{Steiner Tree}, \problem{$k$-Median} and \problem{Facility Location} on graphs of low \emph{doubling dimension} (for a formal definition, see Definition~\ref{def:ddim}). This was improved by Bartal et al.~\cite{Bartal2016}, who obtained a PTAS for \problem{TSP}. As the skeleton dimension of a graph upper bounds its doubling dimension (cf. Section~\ref{sec:skeleton-doubling-dimension}) the aforementioned results immediately imply a PTAS for \problem{TSP} and QPTASs for \problem{Steiner Tree}, \problem{$k$-Median} and \problem{Facility Location}.

The \problem{$k$-Center} problem is NP-complete on general graphs~\cite{Vazirani2001} and has been subject to extensive research. In fact, for any $\epsilon > 0$, it is NP-hard to compute a $(2-\epsilon)$-approximation, even when considering only planar graphs~\cite{Plesnik1980}, geometric graphs using $L_1$ or $L_\infty$ distances or graphs of highway dimension $\mathcal{O}(\log^2 \vert V \vert)$\cite{DBLP:journals/algorithmica/Feldmann19}. However, there is a fairly simple $2$-approximation algorithm for general graphs by Hochbaum and Shmoys~\cite{Hochbaum1986}.

One way to approximate \problem{$k$-Center} better than by a factor of $2$ is the use of so called \emph{fixed-parameter approximation algorithms (FPAs)}. The basic idea is to combine the concepts of fixed-parameter algorithms and approximation algorithms. Formally, for $\alpha > 1$, an $\alpha$-FPA for a parameter $p$ is an algorithm that computes an $\alpha$-approximation in time $f(p) \cdot n^{\mathcal{O}(1)}$ where $f$ is a computable function. Feldmann~\cite{DBLP:journals/algorithmica/Feldmann19} showed there is a $\nicefrac{3}{2}$-FPA for \problem{$k$-Center} when parameterizing both by the number of center nodes $k$ and the highway dimension $hd$. Later, Becker et al.\ \cite{Becker2018} showed that for any $\epsilon > 0$ there is a $(1+\epsilon)$-FPA for $k$-Center when parameterizing by $k$ and $hd$, using a slightly different definition for the highway dimension as in \cite{DBLP:journals/algorithmica/Feldmann19} (see also Section~\ref{sec:highway-dimension}). Moreover, on graphs of doubling dimension $d$, it is possible to compute a $(1 + \epsilon)$-approximation in time $\left(k^k / \epsilon^{\mathcal{O}(k\cdot d)}\right) \cdot n^{\mathcal{O}(1)}$~\cite{Fel18}. As the doubling dimension is a lower bound for the skeleton dimension $\kappa$, this implies a $(1+\epsilon)$-FPA for parameter $(\epsilon, k, \kappa)$. However, computing a $(2-\epsilon)$-approximation is $W[2]$-hard when parameterizing only by $k$, and unless the exponential time hypothesis (ETH) fails, it is not possible to compute a  $(2-\epsilon)$-approximation in time $2^{2^{o(\sqrt{hd})}} \cdot n^{\mathcal{O}(1)}$ for highway dimension $hd$~\cite{DBLP:journals/algorithmica/Feldmann19}.

\subsection{Contributions and Outline}
We first give an overview of various graph parameters, in particular we review several slightly different definitions of the highway dimension that can be found in the literature. Then we show relationships between skeleton dimension, highway dimension and other important parameters. Our results include the following.
\begin{itemize}
		\item The max leaf number $ml$ is a tight upper bound for the bandwidth $bw$. This improves a result of Sorge et al.\ who showed that $bw \leq 2 ml$~\cite{Sorge2019}.
		\item The skeleton dimension is incomparable to any of the parameters distance to linear forest, bandwidth, treewidth and highway dimension (when using the definitions from \cite{abr10} or \cite{abr11}).
		\item The skeleton dimension $\kappa$ is upper bounded by the max leaf number. Moreover, for any graph on at least $3$ vertices there are edge weights for which both parameters are equal. As the max leaf number is an upper bound for the pathwidth $pw$, it follows that $\kappa \geq pw$. This improves a result of Blum and Storandt, who showed that one can choose edge weights for any graph such that the skeleton dimension is at least $(pw - 1) / (\log_2 \vert V \vert + 2)$~\cite{blu18}.
\end{itemize}
The resulting parameter hierarchy is illustrated in Figure~\ref{fig:diagram}. In the second part of the paper we show hardness for two problems in transportation networks.
\begin{itemize}
	\item We show that computing the highway dimension is NP-hard when using the most recent definition from \cite{abr16}. This answers an open question stated in \cite{Fel15b}, where NP-hardness was only shown for the definitions used in \cite{abr10} and \cite{abr11}.
	\item We study the \problem{$k$-Center} problem in graphs of low skeleton dimension. We extend a result from \cite{DBLP:journals/algorithmica/Feldmann19} and show how graphs of low doubling dimension can be embedded into graphs of low skeleton dimension. It follows that for any $\epsilon > 0$ it is NP-hard to compute a $(2 - \epsilon)$-approximation on graphs of skeleton dimension $\mathcal{O}(\log^2 \vert V \vert)$.
\end{itemize}

\tikzstyle{para}=[rectangle,draw=black,minimum height=.8cm,fill=gray!10,rounded corners, on grid]
\tikzstyle{ref}=[rectangle,draw=black,fill=white,rounded corners=.5mm]
\tikzstyle{strict}=[very thick]
\tikzstyle{offset}=[very thick, dashed]
\tikzstyle{incomp}=[thick, dotted]
\definecolor{mygreen}{HTML}{7CCA89}
\tikzstyle{new}=[mygreen]
\newcommand{\tworows}[2]{\begin{tabular}{c}{#1}\\{#2}\end{tabular}}
\newcommand{\distto}[1]{\tworows{Distance to}{#1}}
\begin{figure}[p]
\begin{subfigure}[t]{\textwidth}
\centering
\begin{tikzpicture}[node distance=1.75cm and 3.5cm]
  \node[para] (ml) {Max Leaf \#};
  \node[para] (hd1) [right =7 of ml] {\tworows{Highway}{Dimension 1}};

  \node[para] (dl) [below left=of ml] {\distto{Linear Forest}}
  edge[strict] (ml);
  \node[para] (bw) [below =of ml] {Bandwidth} 
  edge[strict,new] (ml)
  edge[incomp, new] (dl);
  \node[para] (k) [below right=of ml] {\tworows{Skeleton}{Dimension}} 
  edge[strict,new] (ml)
  edge[incomp, new] (bw)
  edge[incomp, new, bend right=20] (dl);
  \node[para] (hd2) [right =of k] {\tworows{Highway}{Dimension 2}}
  edge[strict] (hd1)
  edge[incomp, new] (k)
  edge[incomp, new, bend right=20] (bw)
  edge[incomp, new, bend right=20] (dl);

  \node[para] (pw) [below =of bw] {Pathwidth} 
  edge[offset] (dl)
  edge[strict] (bw);
  \node[para] (md) [below =of k] {\tworows{Maximum}{Degree}} 
  edge[offset] (bw)
  edge[strict] (k);

  \node[para] (tw) [below =of pw] {Treewidth} 
  edge[strict] (pw)
  edge[incomp, new] (k);
  \node[para] (hi) [below =of md] {$h$-index} 
  edge[strict] (md)
  edge[incomp,new] (pw)
  edge[incomp] (tw);

  \node[para] (ac) [below =of tw, xshift=1.5cm] {\tworows{Acyclic}{Chromatic \#}}
  edge[strict] (hi)
  edge[offset] (tw);

  \node[para] (mind) [below =of ac] {\tworows{Minimum}{Degree}}
  edge[strict] (ac);
  \draw (mind.east) edge [incomp, out=0, in=290] (hd1.east);

  \node (strict) [left=27mm of ac,yshift=-0mm] {strict bound};
  \draw [strict] ([xshift=15mm] strict.center) -- ([xshift=25mm] strict.center);
  \node (offset) [below=2mm of strict] {general bound};
  \draw [offset] ([xshift=15mm] offset.center) -- ([xshift=25mm] offset.center);
  \node (incomp) [below=2mm of offset] {incomparable};
  \draw [incomp] ([xshift=15mm] incomp.center) -- ([xshift=25mm] incomp.center);
\end{tikzpicture}
\caption{Relationships between general graph parameters.}\label{fig:diagram_general}
\end{subfigure}

\vspace{1cm}

\begin{subfigure}[t]{\textwidth}
\centering
\begin{tikzpicture}[node distance=1.25cm and 3cm]
  \node[para] (delta) {$\Delta$};
  \node[para] [left of= delta] (ddim) {$ddim$};

  \node[para] (k) [above of=delta] {$\kappa$}
  edge[strict] (delta)
  edge[offset] (ddim);
  \node[para] (hd2) [right of=delta] {$hd_2$}
  edge[incomp, new] (k)
  edge[incomp, bend right=30] (ddim);

  \node[para] (hd3) [above of=k] {$hd_3$}
  edge[strict] (k)
  edge[strict] (hd2);
  \node[para] (hd1) [above =2.5 of hd2] {$hd_1$}
  edge[incomp] (delta)
  edge[incomp,bend right=20] (ddim)
  edge[strict] (hd2);

  \node[para] (hd2g) [above of=hd3] {$\widetilde{hd}_2$}
  edge[strict] (hd3);
  \node[para] [above of=hd1, xshift=15mm] {$hd_3(hd_3+1)$}
  edge[strict] (hd1);

  \node[para] (2hd3) [above of=hd2g] {$2 hd_3$}
  edge[strict] (hd2g);
  \node[para] [left of=2hd3, xshift=-7.5mm] {$(\Delta+1)hd_2$}
  edge[strict] (hd2g);

  \node[para] [right of=2hd3] (hd1g) {$\widetilde{hd}_1$}
  edge[strict] (hd2g)
  edge[strict] (hd1);

  \node[para] [above of=hd1g, xshift=-15mm] {$(\Delta+1)hd_1$}
  edge[strict] (hd1g);
  \node[para] [above of=hd1g, xshift=15mm] {$2hd_3(hd_3+1)$}
  edge[strict] (hd1g);
\end{tikzpicture}
\caption{Relationships between maximum degree $\Delta$, doubling dimension $ddim$, skeleton dimension $\kappa$ and different highway dimensions.}\label{fig:diagram_network}
\end{subfigure}
\caption{Relationships between graph parameters. New results are highlighted in green. Solid lines denote strict bounds (e.g.\ treewidth $\leq$ pathwidth), dashed lines denote general bounds (e.g.\ pathwidth $\leq$ distance to linear forest + 1). Dotted lines denote incomparabilities.}\label{fig:diagram}
\end{figure}
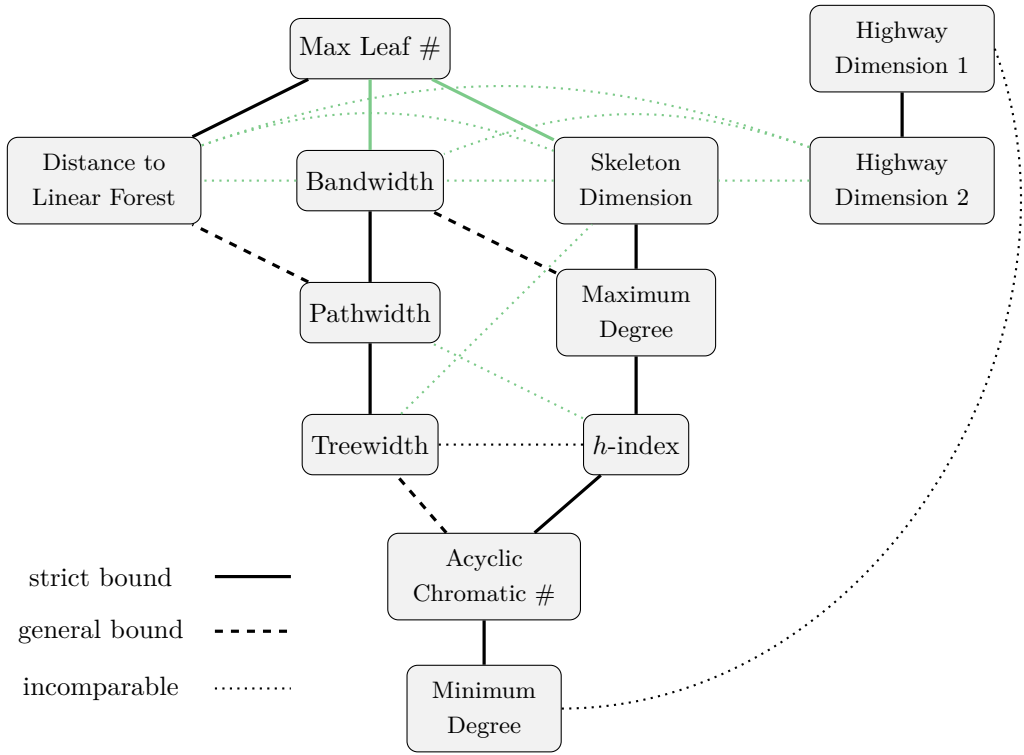
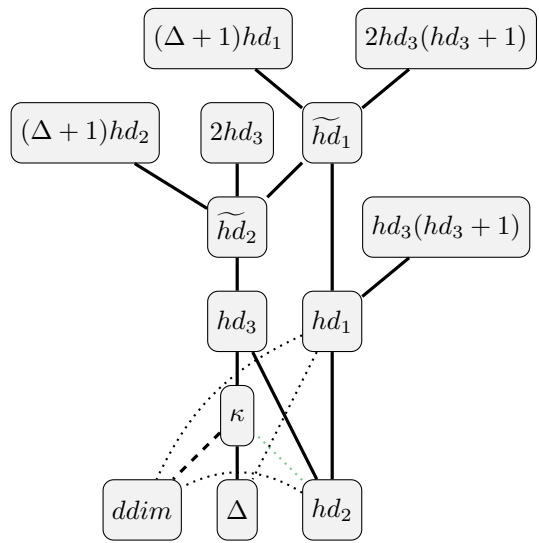

\section{Preliminaries}
We consider undirected graphs $G=(V,E)$ and denote the number of nodes and edges by $n$ and $m$, respectively. Let $\Delta$ be the maximum degree of $G$. For weighted graphs, let $\ell \colon E \rightarrow \mathbb{Q}^+$ be the cost function.  For nodes $u,v \in V$, let $\dist_G(u,v)$ (or simply $\dist(u,v)$) be length of the shortest path from $u$ to $v$ in $G$. A weighted graph $G = (V,E)$ is metric if $(V,\dist_G)$ is a metric, i.e.\ its edge weights satisfy the triangle inequality, that is for all nodes $u,v,w \in V$ we have $\dist(u,w) \leq \dist(u,v) + \dist(v,w)$.  We assume that the shortest path between any two nodes of $G$ is unique, which can be achieved e.g.\ by slightly perturbing the edge weights. For $u \in V$ and $r \in \mathbb{R}$, we define the ball around the node $u$ of radius $r$ as $B_r(u) = \{v \in V \mid \dist(u,v) \leq r \}$. The length of a path $\pi$ is denoted by $\vert \pi \vert$.

\subsection{Skeleton Dimension and Doubling Dimension}\label{sec:skeleton-doubling-dimension}
The skeleton dimension was introduced by Kosowski and Viennot to analyze the performance of hub labels, a route planning technique used for road networks~\cite{kos16}. To define it formally, we first need to introduce the geometric realization $\tilde G=(\tilde V, \tilde E)$ of a graph $G=(V,E)$ with edge weights $\ell$. Intuitively, $\tilde G$ is a continuous version of $G$, where every edge is subdivided into infinitely many infinitely short edges.
This means that $V \subseteq \tilde V$, for all $u,v \in V$ we have $\dist_{\tilde G} (u,v) = \dist_G(u,v)$ and for every edge $\{u,v\}$ of $G$ and every $0 \leq \alpha \leq \ell(\{u,v\})$ there is a node $w \in \tilde V$ satisfying $\dist(u,w) = \alpha$ and $\dist(w,v) = \ell(\{u,v\}) - \alpha$.

For a node $s \in V$ let $T_s$ be the shortest path tree of $s$ and let $\tilde{T_s}$ be its geometric realization. Recall that shortest paths are unique, and hence the same holds for $T_s$ and $\tilde{T_s}$. The skeleton $T^*_s$ is defined as the subtree of $\tilde{T_s}$ induced by the nodes $v \in \tilde V$ that have a descendant $w$ in $\tilde{T_s}$ satisfying $\dist(v,w) \geq \nicefrac{1}{2} \cdot \dist(s,v)$. Intuitively, we obtain $T^*_s$ by taking every shortest path with source $s$, cutting off the last third of the path and taking the union of the truncated paths. For a radius $r \in \mathbb{R}$ let $\mathrm{Cut}_s^r$ be the set of all nodes $u$ in $T^*_s$ satisfying $\dist(s,u) = r$.

\begin{definition}[Skeleton Dimension]
The skeleton dimension $\kappa$ of a graph $G$ is $\max_{s, r} \vert \mathrm{Cut}_s^r \vert$.
\end{definition}

Intuitively, a graph has low skeleton dimension, if for any starting node $s$ there are only a few main roads that contain the major central part of ever shortest path originating from $s$. %
Clearly, the skeleton dimension can be computed in polynomial time by computing the shortest path tree and its skeleton for every node $s \in V$ and determining $\mathrm{Cut}_s^r$ for every radius $r \in \mathbb{R}$. On large networks, a na\"ive implementation is still impracticable, but in~\cite{blu18} it was shown that it is possible to compute $\kappa$ even for networks with millions of vertices.

Related to the skeleton dimension is the doubling dimension, which was introduced as a generalization of several kinds of metrics, e.g.\ Euclidean or Manhattan metrics.

\begin{definition}[Doubling Dimension]\label{def:ddim}
A graph $G$ is $d$-doubling, if for any radius $r$, any ball of radius $r$ is contained in the union of $d$ balls of radius $\nicefrac{r}{2}$.
If $d$ is the smallest such integer, the doubling dimension of $G$ is $\log_2 d$.
\end{definition}

Computing the doubling dimension is NP-hard~\cite{Gottlieb2013}. Kosowski and Viennot showed that a graph with skeleton dimension $\kappa$ is $(2\kappa + 1)$ doubling \cite{kos16}.

\subsection{Highway Dimension}\label{sec:highway-dimension}

The highway dimension was introduced by Abraham et al.,\ motivated by the observation of Bast et al. that in road networks, all shortest paths leaving a certain region pass through one of a small number of nodes \cite{bas06,bas07}. In the literature, several slightly different definitions of the highway dimension can be found. The first one was given in \cite{abr10}.

\begin{definition}[Highway Dimension 1]\label{def:hd1}
The highway dimension of a graph $G$ is the smallest integer $hd_1$ such that for any radius $r$ and any node $u$ there is a hitting set $S \subseteq B_{4r}(u)$ of size $hd_1$ for the set of all shortest paths $\pi$ satisfying $\vert \pi \vert > r$ and $\pi \subseteq B_{4r}(u)$.
\end{definition}

In \cite{DBLP:journals/siamcomp/FeldmannFKP18,Fel18}, a generalized version of $hd_1$ was used, where balls of radius $c \cdot r$ for $c \geq 4$ were considered. It was observed that the highway dimension is highly sensitive to the chosen radius, i.e.\ there are graphs of highway dimension $1$ w.r.t.\ radius $c$ and highway dimension of $\Omega(n)$ w.r.t.\ radius $c' > c$.

In \cite{abr11} the highway dimension was defined as follows.

\begin{definition}[Highway Dimension 2]
The highway dimension of a graph $G$ is the smallest integer $hd_2$ such that for any radius $r$ and any node $u$ there is a hitting set $S \subseteq V$ of size $hd_2$ for the set of all shortest paths $\pi$ satisfying $2r \geq \vert \pi \vert > r$ that intersect $B_{2r}(u)$.
\end{definition}

The definition of $hd_1$ requires to hit all shortest paths contained in the ball of radius $4r$, while for $hd_2$ only the shortest paths intersecting the ball of radius $2r$ need to be hit. Hence, we have $hd_2 \leq hd_1$. Abraham et al.\ motivate their new definition with the fact that a smaller highway dimension can be achieved on real-world instances, while previous results still hold \cite{abr11}. Both previously defined highway dimensions are incomparable to the maximum degree and the doubling dimension \cite{abr10}.

In \cite{abr16}, a continuous version of the highway dimension $hd_2$ was introduced, which is based on the geometric realization. For the definition, assume w.l.o.g.\ that $\ell(e) \geq 1$ for all edges $e \in E$.

\begin{definition}[Continuous Highway Dimension]
The continuous highway dimension of a graph $G$ is the smallest integer $\widetilde{hd}_2$ such that for any radius $r \geq 1$ and any node $u \in \tilde V$ of the geometric realization $\tilde G$ there is a hitting set $S \subseteq V$ of size $\widetilde{hd}_2$ for the set of all shortest paths $\pi$ satisfying $2r \geq \vert \pi \vert > r$ that intersect $B_{2r}(u)$.
\end{definition}

Clearly, we have $hd_2 \leq \widetilde{hd}_2$. In \cite{kos16} it was observed that $\widetilde{hd}_2$ is upper bounded by $(\Delta + 1) hd_2$. Along the lines of Definition~\ref{def:hd1}, we can also introduce the continuous version $\widetilde{hd}_1$ of $hd_1$. It holds that $hd_1 \leq \widetilde{hd}_1 \leq (\Delta + 1) hd_1$ and moreover $\widetilde{hd}_2 \leq \widetilde{hd}_1$.
In \cite{abr16}, yet another definition of the highway dimension was given. It is based on the notion of $r$-significant shortest paths.

\begin{definition}[$r$-significant shortest path]
		For $r \in \mathbb{R}$, a shortest path $\pi = v_1 \dots v_k$ is $r$-significant iff it has an $r$-witness path $\pi'$, which means that $\pi'$ is a shortest path satisfying $\vert \pi' \vert > r$ and one of the following conditions hold: (i) $\pi' = \pi$, or (ii) $\pi' = v_0 \pi$, or $\pi' = \pi v_{k+1}$, or (iv) $\pi' = v_0 \pi v_{k+1}$ for nodes $v_0, v_{k+1} \in V$.
\end{definition}

In other words, $\pi$ is $r$-significant, if by adding at most one vertex to every end we can obtain a shortest path $\pi'$ of length more than $r$ (the $r$-witness). For $r,d \in \mathbb{R}$, a shortest path $\pi$ is $(r,d)$-close to a vertex $v$, if there is an $r$-witness path $\pi'$ of $\pi$ that intersects the ball $B_d(v)$.

\begin{definition}[Highway Dimension 3]
The highway dimension of a graph $G$ is the smallest integer $hd_3$ such that for any radius $r$ and any node $u$ there is a hitting set $S \subseteq V$ of size $hd_3$ for the set of all shortest paths $\pi$ that are $(r,2r)$-close to $u$.
\end{definition}

The advantage of the latest definition is that it also captures continuous graphs. In particular, it was shown that $hd_3 \leq \widetilde{hd}_2 \leq 2 hd_3$~\cite{abr16}. Hence there is no need for a continuous version of $hd_3$, apart from the fact that there is no meaningful notion of an $r$-witness in a continuous graph.

It can be easily seen that $hd_2 \leq hd_3$ as every shortest path $\pi$ that is longer than $r$ and intersects $B_{2r}(u)$ is also $(r,2r)$-close to $u$ (using $\pi$ itself as the $r$-witness). Moreover, the skeleton dimension $\kappa$ is a lower bound for $hd_3$, i.e.\ $\kappa \leq hd_3$ \cite{kos16}. Feldmann et al.\ showed that $hd_1 \leq hd_3(hd_3 + 1)$ \cite{Fel15b}. Combining their proof with \cite{abr16} yields that $\widetilde{hd}_1 \leq 2 hd_3(hd_3 + 1)$.

Computing the highway dimensions $hd_1$ and $hd_2$ is NP-hard~\cite{Fel15b}. In Section~\ref{sec:Hardness-hd3} we show that this also holds for $hd_3$, which answers an open question stated in~\cite{Fel15b}.

\subsection{Classic graph parameters}
We now provide an overview of several classic graph parameters. They are all defined on unweighted graphs, but we can also apply them to weighted graphs, simply neglecting edge weights.
We start with introducing the treewidth and the related parameters pathwidth and bandwidth.

\begin{definition}[Treewidth]
A tree decomposition of a graph $G = (V,E)$ is a tree $T = (\mathcal{X}, \mathcal{E})$ where every node (also called \emph{bag}) $X \in \mathcal{X}$ is a subset of $V$ and the following properties are satisfied: (i) $\bigcup_{X \in \mathcal{X}} X = V$, (ii) for every edge $\{u,v\} \in E$ there is a bag $X \in \mathcal{X}$ containing both $u$ and $v$, and (iii) for every $u \in V$, the set of all bags containing $u$ induce a connected subtree of $T$. The width of a tree decomposition $T = (\mathcal{X}, \mathcal{E})$ is the size of the largest bag minus one, i.e. $\max_{X \in \mathcal{X}} \left( \vert X \vert - 1 \right)$. The treewidth $tw$ of a graph $G = (V,E)$ is defined as the minimum width of all tree decompositions of $G$.
\end{definition}

\begin{definition}[Pathwidth]
A path decomposition of a graph $G$ is a tree decomposition of $G$ that is a path. The pathwidth $pw$ of $G$ is the minimum width of all path decompositions of $G$.
\end{definition}

It follows directly from the definitions, that the pathwidth is an upper bound for the treewidth and one can show that the minimum degree is a lower bound for the treewidth~\cite{Sorge2019}. The maximum degree $\Delta$ is incomparable to both treewidth and pathwidth, as for a square grid graph we have $\Delta = 4$ and $tw \in \Omega(\sqrt n)$ whereas for a star graph we obtain $\Delta \in \Omega(n)$ and $pw = 1$.

\begin{definition}[Bandwidth]
A vertex labeling of a graph $G = (V,E)$ is a bijection $f\colon V \rightarrow \{1, \dots, n\}$. The bandwidth of $G$ is the minimum of $\max\{\vert f(u) - f(v)\vert \colon \{u,v\} \in E\}$, taken over all vertex labelings $f$ of $G$.
\end{definition}

It was shown that the bandwidth $bw$ is a tight upper bound for the pathwidth~\cite{kap96}, and that $\Delta \leq 2 \cdot bw$~\cite{Sorge2019}.

\begin{definition}[Max Leaf Number]
The max leaf number $ml$ of a graph $G$ is the maximum number of leaves of all spanning trees of $G$.
\end{definition}

\begin{definition}[Distance to Linear Forest]
The distance to linear forest (also known as distance to union of paths) of a graph $G = (V,E)$ is the size of the smallest set $S \subseteq V$ that separates $G$ into a set of disjoint paths.
\end{definition}

\begin{definition}[$h$-Index]
The $h$-index of a graph $G = (V,E)$ is the largest integer $h$ such that $G$ has $h$ vertices of degree at least $h$.
\end{definition}

The max leaf number is closely related to the notion of a connected dominating set. It is an upper bound for several graph parameters. It was shown that for the max leaf number $ml$ and the distance to linear forest $dl$ we have $dl \leq ml-1$~\cite{DeLaVina2009}. We will show that it also upper bounds the bandwidth and the skeleton dimension. For distance to linear forest $dl$ and pathwidth $pw$ it is known that $pw \leq dl + 1$~\cite{Bodlaender1998}.
Clearly, the $h$-index is a lower bound for the maximum degree. It was shown that the $h$-index is incomparable to the treewidth~\cite{Sorge2019}.

\section{Parameter Relationships}
In this section we show relationships between skeleton dimension, highway dimension and other graph parameters. We will see that the max leaf number is an upper bound for the skeleton dimension and the bandwidth, whereas many of the remaining parameters are pairwise incomparable. This shows that they are all useful and worth studying.

\subsection{Upper Bounds}
We first relate the max leaf number to the skeleton dimension and the bandwidth. We will use the fact, that every tree has as least as many leaves as any subtree.%

\begin{observation}\label{obs:subtree}
Let $T'$ be a subtree of a tree $T$ and let $L$ and $L'$ be the leaves of $T$ and $T'$, respectively. Then we have $\vert L' \vert \leq \vert L \vert$.
\end{observation}

This allows to show that the max leaf number is an upper bound for the skeleton dimension.

\begin{theorem}\label{k-ml}
For the skeleton dimension $\kappa$ and the max leaf number $ml$ we have $\kappa \leq ml$. 
For any unweighted undirected graph on $n \geq 3$ nodes there are metric edge weights such that $\kappa = ml$.
\end{theorem}

\begin{proof}
		Let $G = (V,E)$ be a graph. Consider the skeleton $T_s^*$ of some node $s \in V$ that has a cut $C$ of size $\kappa$. As for any two distinct nodes $u,v \in C$ the lowest common ancestor in $T_s^*$ is distinct from $u$ and $v$, $T_s^*$ has at least $\kappa$ leaves. The skeleton $T_s^*$ is a subtree of the shortest path tree $T_s$ of $s$, so Observation~\ref{obs:subtree} implies that $T_s$ has at least $\kappa$ leaves. As $T_s$ is a spanning tree of $G$ it follows that $\kappa \leq ml$.

To show that the bound is tight, consider a spanning tree $T = (V, E_T)$ of an unweighted graph $G = (V,E)$ with $ml$ leaves. We choose edge weights $\ell$ such that the skeleton dimension of the resulting weighted graph equals $ml$. Let
\[
	\ell(\{u,v\}) = \begin{cases}
		2 & \text{if } \{u,v\} \in E_T \text{ and } u \text{ or } v \text{ is a leaf of } T\\
		\nicefrac{1}{n} & \text{if } \{u,v\} \in E_T \text{ and neither } u \text{ nor } v \text{ is a leaf of } T\\
		5 & \text{else}
	\end{cases}
\]
To examine the skeleton dimension of the resulting graph, consider an internal node $s$ of $T$. Such a node exists if $n > 2$. We observe that the shortest path tree $T_s$ of $s$ is equal to $T$ as for any vertex $v$ we have $\dist(s,v) < 3$, and hence no edge $e \in E \setminus E_T$ can be contained in $T_s$. Moreover, for any leaf $v$ we have $\dist(s,v) \geq 2$ and for any internal node $v$ we have $\dist(s,v) < 1$. Consider now the skeleton $T_s^*$.  Any leaf of $T_s^*$ has distance at least $\nicefrac{2}{3} \cdot 2 > 1$ from $s$.
As $T_s^*$ has $ml$ leaves, the cut of $T_s^*$ at radius $\nicefrac{4}{3}$ has size $ml$.

Note that in general, the resulting graph is not metric. To fix this, let $\dist_T(u,v)$ be the shortest path distance from $u$ to $v$ when applying the previously chosen edge weights. For $\{u,v\} \in E_T$ we define $\ell$ as previously, but for $\{u,v\} \not \in E_T$ choose $\ell(u,v) = \dist_T(u,v) -\epsilon$ where for every edge, $\epsilon$ is chosen from $(0,\nicefrac{1}{n^2})$ such that shortest paths are unique. 
Consider an internal node $s$ of $T$. The shortest path tree $T_s$ of $s$ may now differ from $T$, but the number of leaves of $T_s$ is still $ml$. For any leaf $v$ of $T$ we have now $\dist(s,v) > 2 - \nicefrac{n}{n^2} \geq \nicefrac{3}{2}$ and for any internal node $v$ we have $\dist(s,v) < 1$. Hence, the cut of $T_s^*$ at radius $1$ has size $ml$.
\end{proof}

As the max leaf number $ml$ is an upper bound for the pathwidth $pw$, it follows that for any graph $G$ on $n \geq 3$ nodes there are edge weights such that $\kappa \geq pw$. This improves a result of Blum and Storandt, who showed that there are edge weights such that $\kappa \geq (pw-1) / (\log_2 n + 2)$~\cite{blu18}.

Sorge et al.\ showed that the bandwidth can be upper bounded by two times the max leaf number~\cite{Sorge2019}. We slightly modify their proof to remove the factor of $2$ and show that the resulting bound is tight.

\begin{lemma}\label{bw-ml}
For the max leaf number $ml$ and the bandwidth $bw$ we have $bw \leq ml$. This bound is tight.
\end{lemma}

\begin{proof}
Let $T$ be a BFS tree of a graph $G = (V,E)$ and let $f\colon V \rightarrow \{1, \dots, n\}$ be a vertex labeling that assigns to every node the time of its BFS discovery. W.l.o.g.\ we assume that $f(v_i) = i$. Choose an edge $\{v_i, v_j\} \in E$ maximizing $f(v_j) - f(v_i)$. It follows that $bw \leq f(v_j) - f(v_i) = j-i$.
		
Observe that in the BFS tree $T$, the node $v_i$ is the parent of $v_j$ as by the choice of $\{v_i, v_j\}$ there is no $k < i$ such that $\{v_k, v_j\} \in E$. Consider the subtree $T'$ of $T$ induced by the nodes $\{v_1, \dots, v_j\}$. As $v_i$ is the parent of $v_j$ and nodes are ordered by their discovery time, it follows that $v_{i+1}, \dots, v_j$ are leaves of $T'$. Observation~\ref{obs:subtree} implies $T$ has at least $(j - i)$ leaves.

Tightness follows from the complete graph $K_n$ where $bw = ml = n-1$.
\end{proof}

\subsection{Incomparabilities}
We now show incomparabilities between several parameters, which means that they are all worth studying.
In \cite{Sorge2019} it was proven that the treewidth is incomparable to the $h$-index. We observe that the same holds for the pathwidth.
\begin{theorem}
The pathwidth and $h$-index are incomparable.
\end{theorem}

\begin{proof}
The $\sqrt n \times \sqrt n$ grid graph has pathwidth $\sqrt n$ and $h$-index at most $4$. The caterpillar tree with $d$ backbone vertices of degree $d$ has pathwidth $1$ and $h$-index $d$.
\end{proof}

We proceed with relating the highway dimensions $hd_1$ and $hd_2$ to the treewidth and pathwidth. In~\cite{DBLP:journals/siamcomp/FeldmannFKP18} it was observed that graphs of low highway dimension $hd_1$ do not have bounded treewidth, as the complete graph on vertex set $\{1,\dots,n\}$ with edge weights $\ell(\{i,j\}) = 4^{\max(i,j)}$ has highway dimension $hd_1 = 1$ and treewidth $n-1$.\footnote{The edge weights chosen in~\cite{DBLP:journals/siamcomp/FeldmannFKP18} are actually $\ell(\{i,j\}) = 4^{\min(i,j)}$, which results in a non-metric graph. Removing all edges that are not a shortest path yields a star graph of treewidth $1$.} 
The complete graph $K_n$ has indeed a minimum degree of $n-1$, which is a lower bound for the treewidth.
On the other hand, there are graphs of constant bandwidth and a linear highway dimension $hd_2$. For instance, consider a complete caterpillar tree on $b$ backbone vertices of degree $3$. Its bandwidth is $2$. Choose the weight of an edge as $\nicefrac{1}{n}$ if it is a backbone edge and as $1$ otherwise. Every edge of weight $1$ is a shortest path intersecting the ball of radius $1$ around some fixed backbone vertex and hence $hd_2 \geq b = \nicefrac{n}{2}-2$. This gives us the follows theorem.

\begin{theorem}\label{thm:hd-bw}
The highway dimensions $hd_1$ and $hd_2$ are incomparable to the bandwidth and the minimum degree.
\end{theorem}

We would also like to relate the skeleton dimension to bandwidth and treewidth. On general graphs, it is easy to show, that the skeleton dimension is incomparable to the other two parameters. For instance, a star graph has treewidth $1$ and linear skeleton dimension, whereas a complete graph has linear treewidth, but we can choose edge weights such that the shortest path tree of every vertex becomes a path which implies a constant skeleton dimension. However, by choosing such weights for the latter graph, most edges become useless as they do not represent a shortest path and removing all unnecessary edges produces a graph of low treewidth. Still, we can show, that even on metric graphs the skeleton dimension is incomparable to both bandwidth and treewidth.

\begin{theorem}
On metric graphs the skeleton dimension and the bandwidth are incomparable.
\end{theorem}

\begin{proof}
Consider the complete caterpillar tree on $b$ backbone vertices of degree $3$. It has a bandwidth of $2$. Set the weight of every backbone edge to $1$ and pick an arbitrary backbone vertex $v$. For the remaining edges, choose edge weights such that all leaves have the same distance $d \geq 2$ from $v$. It follows that the skeleton dimension of the weighted caterpillar tree equals the number of leaves which is $b + 2 = \nicefrac{n}{2}+1$.

The complete binary tree $B_{2d + 1}$ of depth $2d+1$ has pathwidth $d$ \cite{cat96}. We show that there are edge weights for $B_{2d + 1}$ such that the skeleton dimension is at most $3$. Let $s$ be the root of $B_{2d + 1}$.  We call the depth of a vertex in the tree also its \emph{level} and choose the weight of an edge $\{v,w\}$ as $\ell(\{v,w\}) = 3^{-j}$ if $v$ and $w$ are level $j$ and level $(j+1)$ vertices, respectively.

Let $v$ be a level $i$ vertex.  We show that for any radius $r$ we have $\vert \mathrm{Cut}_v^r \vert \leq 3$.
Clearly the shortest path $\pi$ form $v$ to $s$ is contained in the skeleton $T^*_v$ of the shortest path tree $T_v$ as the root $s$ has a descendant $w$ satisfying $\dist(s,w) \geq \nicefrac{1}{2} \cdot \dist(v,s)$. For $0 \leq j \leq i$ let $v_j$ be the unique level $j$ vertex on the path $\pi$ and let $w$ be a descendant of some $v_j$ such that the shortest path from $v_j$ to $w$ is edge-disjoint from $\pi$. 
Assume that the vertex $w$ is contained in the skeleton $T^*_v$. This means that $w$ has a descendant $w'$ such that $\dist(w,w') \geq \nicefrac{1}{2} \cdot \dist(v,w)$. As moreover $\dist(v_j,w) = \dist(v_j,w') - \dist(w,w')$ and $\dist(v_j,w) \leq \dist(v,w)$, it follows that $\nicefrac{3}{2} \cdot \dist(v_j,w) \leq \dist(v_j,w')$ which implies $\dist(v_j,w') \leq \nicefrac{2}{3} \sum_{x=j}^{2d+1} 3^{-x} < \nicefrac{2}{3} \sum_{x=j}^\infty 3^{-x} = \frac{2}{3} \cdot \frac{3^{-j+1}}{2} = 3^{-j}$.

To bound the size of $\mathrm{Cut}_v^r$, consider a radius $r > 0$ and let $y$ be the node in shortest path from $v$ to the root $s$ that maximizes $\dist(v,y)$ while satisfying $r' := \dist(v,y) \leq r$. Let $j$ be the level of $y$. From our previous observation it follows that $y$ is the only vertex that has distance $r'$ from $v$ and is contained in the skeleton $T^*_v$. Moreover, there are at most three vertices at distance $r - r'$ from $y$. It follows that $\vert \mathrm{Cut}_v^r \vert \leq 3$ and that $B_{2d + 1}$ has skeleton dimension $\kappa \leq 3$.
\end{proof}

\begin{theorem}
On metric graphs the skeleton dimension and the treewidth are incomparable.
\end{theorem}

\begin{proof}
The star graph $S_n$ on $n$ vertices has treewidth $1$ and skeleton dimension $n-1$.

We now construct a graph of treewidth $\Omega(\sqrt n)$ and constant skeleton dimension. Consider a square grid graph $G$ on the vertex set $V = \{v_1, \dots, v_n\}$. Subdivide every edge $\{u,v\}$ by inserting two vertices $x_{uv}$ and $y_{uv}$, i.e.\ replace the edge $\{u,v\}$ through a path $u\,x_{uv}\,y_{uv}\,v$. Connect the vertices $v_1, \dots, v_n$ through a path $P$ and denote the resulting graph by $G' = (V', E')$. The original grid graph $G$ has treewidth $\sqrt n$ and is a minor of $G'$. Hence, $G'$ has treewidth $\Omega(\sqrt n)$.

We now choose edges weights for $G'$ resulting in a constant skeleton dimension. For every edge $e$ that is part of the path $P$, let $\ell(e) = 1$. Consider an edge $\{u,v\}$ of $G$ that was replaced by the path $u\,x_{uv}\,y_{uv}\,v$ and denote the shortest path distance between $u$ and $v$ on the path $P$ by $\dist_P(u,v)$. We choose $\ell(\{u,x_{uv}\}) = \ell(\{y_{uv},v\}) = 1$ and $\ell(\{x_{uv},y_{uv}\}) = \dist_P(u,v) + \nicefrac{1}{2}$. It is easy to verify that the resulting graph is metric.

To bound the skeleton dimension, we use the following claim: For every edge $\{u,v\}$ of $G$, neither of the shortest paths from $u$ to $x_{uv}$ or from $v$ to $y_{uv}$ contains the edge $\{x_{uv},y_{uv}\}$. To prove the claim, observe that by concatenating the subpath of $P$ between $u$ and $v$ and the edge $\{v,y_{uv}\}$, we obtain a path of length $\dist_P(u,v) + 1$. Any path from $u$ to $y_{uv}$ containing the edge $\{x_{uv},y_{uv}\}$ has length $\dist_P(u,v) + \nicefrac{3}{2}$. The case of $v$ and $x_{uv}$ is symmetric.

It follows that in $G'$ the shortest path tree of a vertex $s$ cannot contain the edge $\{x_{uv},y_{uv}\}$ unless $s \in \{x_{uv},y_{uv}\}$, as any subpath of a shortest path must be a shortest path itself.
Consider the shortest path tree $T_s$ of some vertex $s \in V$. The previous claim implies that $T_s$ is a caterpillar tree where $P$ is the backbone path. Moreover, $T_s$ ha maximum degree $\Delta \leq 6$ and all edges have unit length. 
Let $r > 0$ and consider the set $\mathrm{Cut}_s^r$. For $r \leq 1$, the set $\mathrm{Cut}_s^r$ intersects only edges incident to $s$ and hence $\vert \mathrm{Cut}_s^r \vert \leq 6$. For $1 < r \leq 2$, the set $\mathrm{Cut}_s^r$ intersects only edges incident to the two neighbors of $s$ on $P$, which implies $\vert \mathrm{Cut}_s^r \vert \leq 10$. Finally, for $r > 2$ we have $\vert \mathrm{Cut}_s^r \vert \leq 2$ because for any vertex $v \not \in \tilde P$, the distance to its furthest descendant is less than $1 < r/2$ and hence, the set $\mathrm{Cut}_s^r$ intersects only edges from the path $P$. Similarly, it can be shown that $\vert \mathrm{Cut}_s^r \vert \leq 6$ if $s \not \in V$ (i.e. $s = x_{uv}$ or $s = y_{uv}$). It follows that the skeleton dimension of $G'$ is $\kappa \leq 10$.
\end{proof}

So far, it was only known that there can be an exponential gap between skeleton and highway dimension~\cite{kos16}. However, we can use the graph $G'$ from the previous proof to show that the skeleton dimension and the highway dimensions $hd_1$ and $hd_2$ are incomparable. Let $\{v^{1,1}, \dots, v^{q,q}\}$ be the vertex set of the original grid graph and choose the path $P$ used in the construction of $G'$ as $v^{1,1} \dots v^{1,q}\,v^{2,1} \dots v^{2,q} \dots v^{q,1} \dots v^{q,q}$.
In the resulting graph $G'$, for $i \in \{1,\dots,q\}$, the shortest path from $v_{1,i}$ to $v_{2,i}$ has length $q$ and hence the edge $e_i = \{x_{v^{1,i},v^{2,i}},y_{v^{1,i},v^{2,i}}\}$ has length $q+\frac{1}{2}$. As any edge of $\{e_1, \dots, e_q\}$ intersects the ball around $v^{1,1}$ of radius $2q$ and no two of this edges share a common vertex, the highway dimension $hd_2$ of $G'$ is at least $q = \sqrt n$. The star graph on $n$ vertices with unit edge weights has a skeleton dimension of $n-1$ and a highway dimension $hd_1$ of $1$, so we obtain the following corollary.

\begin{corollary}
The skeleton dimension is incomparable to both highway dimensions $hd_1$ and $hd_2$.
\end{corollary}

Finally it can be shown that the distance to linear forest $dl$ is incomparable to the bandwidth $bw$, the skeleton dimension $\kappa$ and the highway dimensions $hd_1$ and $hd_2$. For instance, a caterpillar tree of constant maximum degree has a distance to linear forest of $\Omega(n)$, but constant bandwidth, skeleton dimension and highway dimensions (for suitably chosen edge weights), whereas there are star-like graphs for which $dl \in \mathcal{O}(1)$ and $bw, \kappa, hd_1, hd_2 \in \Omega(n)$.

\begin{theorem}
The distance to linear forest is incomparable to the bandwidth, the skeleton dimension and the highway dimensions $hd_1$ and $hd_2$.
\end{theorem}

\begin{proof}
We will use the fact that the caterpillar tree $\mathcal{C}_b$ on $b$ backbone vertices of degree $3$ has a distance to linear forest of $b = \nicefrac{n}{2}-1$

\textbf{Bandwidth.}
The caterpillar $\mathcal{C}_b$ has bandwidth $2$. The star graph $S_n$ on $n$ vertices has a bandwidth of $\lfloor \nicefrac{n}{2} \rfloor$ and a distance to linear forest of $1$.

\textbf{Skeleton dimension.}
Consider the caterpillar $\mathcal{C}_b$ and choose the weight of an edge $\{u,v\}$ as $2$ if $u$ and $v$ are both backbone vertices and as $1$ otherwise. The skeleton $T_s^*$ of any vertex $s$ contains exactly one vertex of degree $3$ (the backbone vertex that is closest to $s$) and no vertex of degree more than $3$. Hence, the skeleton dimension is $3$.
The star graph $S_n$ on $n$ vertices with unit edge weights has a skeleton dimension of $n-1$ and a distance to linear forest of $1$.

\textbf{Highway dimensions.}
Consider the caterpillar $\mathcal{C}_b$ and choose the weight of an edge $\{u,v\}$ as $5$ if $u$ and $v$ are both backbone vertices and as $1$ otherwise. To bound the highway dimension $hd_1$, consider some node $v$ and let $r > 0$. Consider a maximum path $P \subseteq B_{4r}(v)$ containing only backbone vertices. It holds that $\vert P \vert \leq 8r$. We can greedily choose a set $S \subseteq P$ such that $\vert S \vert \leq 7$ and any subpath $\pi$ of $P$ of length $\vert \pi \vert \geq r-2$ is hit by $S$. Consider a path $\pi' \subseteq B_{4r}(v)$ that is not hit by $S$. The path $\pi'$ contains at most two edges of length $1$ incident to a leaf and a subpath of $P$ that has length less than $r-2$. Hence, $\pi'$ has length at most $r$. It follows that for any $v \in V$ and any $r > 0$ we can hit all shortest paths $\pi$ satisfying $\vert \pi \vert > r$ and $\pi \subseteq B_{4r}(v)$ with at most $7$ vertices, which means that $hd_1 \leq 7$.

Take a star graph with $l$ leaves, subdivide every edge by inserting one vertex and choose the weight of every edge in the resulting graph as $1$. We obtain a graph of distance to linear forest $1$ and highway dimension $hd_2 = l$, as every edge incident to a leaf is a shortest path of length $1$ intersecting the ball of radius $1$ around the central vertex.
\end{proof}

\section{Hardness Results}
In this section we show hardness for two problems in transportation networks. We first show that computing the highway dimension in NP-hard, even when using the most recent definition. Then we consider the \problem{$k$-Center} problem and show that for any $\epsilon > 0$, computing a $(2-\epsilon)$-approximation is NP-hard on graphs of skeleton dimension $\mathcal{O}(\log^2 n)$.

\subsection{Highway Dimension Computation}\label{sec:Hardness-hd3}
In~\cite{Fel15b} it was shown that computing the highway dimension $hd_1$ is NP-hard. The presented reduction is from \problem{Vertex Cover} and also works for $hd_2$. It does not directly carry over to $hd_3$ as the constructed graph has maximum degree $\Delta = n-1$ and we have $hd_3 \geq \Delta$. Still, using a slightly different reduction, we can show NP-hardness for the computation of $hd_3$.

\begin{theorem}
Computing the highway dimension $hd_3$ is NP-hard.
\end{theorem}

\begin{proof}
We present a reduction from \problem{Vertex Cover} on graphs with maximum degree $\Delta \leq 3$. Consider therefore a graph $G = (V, E)$ on $n$ nodes satisfying $\Delta \leq 3$. We construct a weighted graph $G' = (V',E')$ as follows. Add a single node $x$ and for any node $v \in V$, add a new node $v^*$ and the edges $\{v,v^*\}$ and $\{v^*, x\}$. For an edge $e \in E'$ choose edge weight $\ell(e) = 5$ if $e$ is incident to $x$ and $\ell(e) = 1$ otherwise.

Let $C$ be a minimum vertex cover of $G$.
We may assume that $\vert C \vert > ((\Delta+1)^6 - 1)/\Delta \in \mathcal{O}(1)$ as for any constant $c$ we can decide in polynomial time whether $G$ has a minimum vertex cover of size $c$. We show that $G'$ has highway dimension $hd_3 = \vert C \vert + n + 1$. Observe that $hd_3$ is still linear in $n$, but it may vary between $n+1$ and $2n$, depending on $\vert C \vert$.

Let $0 < r < \nicefrac{5}{2}$. Consider a node $u \in V'$. Let $N$ be the closed neighborhood of the ball around $u$ of radius $2r$, i.e.\ $v \in N$ iff $v \in B_{2r}(u)$ or $v$ is adjacent to a node $w \in B_{2r}(u)$. Clearly, $N$ is a hitting set for all shortest paths that are $(r,2r)$-close to $u$.
For $u \neq x$, the ball $B_{2r}(u)$ contains at most $\sum_{i=0}^4 (\Delta + 1)^i$ nodes, as $2r < 5$. Moreover, every node in $B_{2r}(u)$ has at most $\Delta + 1$ neighbors. Hence, $N$ is a hitting set of size $\sum_{i=0}^5 (\Delta + 1)^i = ((\Delta+1)^6 - 1)/\Delta$ for all shortest paths that are $(r,2r)$-close to $u$. For $u = x$, we have $N = V' \setminus V$ and therefore $\vert N \vert = n+1$.

Let $r = \nicefrac{5}{2}$. The ball around $x$ of radius $2r = 5$ is $B_{2r}(x) = V' \setminus V$. Any edge $\{u,v\} \in E$ is $(r,2r)$-close to $x$, as $u^*\,u\,v\,v^*$ is an $r$-witness. Moreover, any node $u \in V' \setminus V$ is a shortest path that is $(r,2r)$-close to $x$. However, a single node $u \in V$ is not $r$-significant, as it can only be extended to a witness of length $2 < r$. Hence, a shortest path $\pi$ is $(r,2r)$-close to $x$ iff and only if $\pi \in E$ or $\pi \in V' \setminus V$. Consider a smallest hitting set $H \subseteq V$ for all shortest paths that are $(r,2r)$-close to $x$. We have $(V' \setminus V) \subseteq H$, we $H$ needs to hit all paths that consist of one single node $v \in V' \setminus V$. Moreover, $H$ needs to hit all edges $E$. In other words, $H$ consists of $V' \setminus V$ and a vertex cover for $G$. Hence, the hitting set $H$ has size $\vert V' \setminus V \vert + \vert C \vert = \vert C \vert + n + 1$.

Observe that for $r=\nicefrac{5}{2}$, any $r$-significant shortest path in $G'$ is $(r,2r)$-close to $x$, as any node of $G'$ has a neighbor contained in $B_{2r}(x)$. Hence, for any node $u \in V'$, there is a hitting set for all shortest paths that are $(r,2r)$-close to $u$ of size at most $\vert C \vert + n + 1$. Moreover, for any node $u$ and any $r > \nicefrac{5}{2}$, a shortest path can only be $(r,2r)$-close to $u$, if it is also $(\nicefrac{5}{2},5)$-close to $u$. Hence, for any $u \in V'$ and any $r > \nicefrac{5}{2}$, for all shortest paths that are $(r,2r)$-close to $u$ there is a hitting set of size at most $\vert C \vert + n + 1$.

We conclude that the highway dimension of $G'$ is $hd_3 = \vert C \vert + n + 1$ if and only if $G$ has a minimum vertex cover of size $\vert C \vert$.
\end{proof}

\subsection{Hardness of Approximating $k$-Center}
In the \problem{$k$-Center} problem, we are given a graph $G = (V,E)$ with positive edge weights and the goal is to select $k$ center nodes $C \subseteq V$ while minimizing $\max_{u \in V} \min_{v \in C} \dist(u,v)$, that is the maximum distance from any node to the closest center node. A possible scenario is that one wants to place a limited number of hospitals on a map such that the maximum distance from any point to the closest hospital is minimized. 

We will prove that computing a $(2-\epsilon)$-approximation on graphs with low skeleton dimension is NP-hard. For that purpose, we first show the following lemma, which is a non-trivial extension of a result of Feldmann \cite{DBLP:journals/algorithmica/Feldmann19}. The aspect ratio of a metric $(X, \dist_X)$ is the ratio of the maximum distance between any pair of vertices in $X$ and the minimum distance.

\begin{lemma}\label{lem:distortion}
Let $(X, \dist_X)$ be a metric of constant doubling dimension $d$ and aspect ratio $\alpha$. For any $0 < \epsilon < 1$ it is possible to compute a graph $G = (X,E)$ in polynomial time that has the following properties:
\begin{alphaenumerate}
		\item \label{prop:distortion} for all $u,v \in X$ we have $\dist_X(u,v) \leq \dist_G(u,v) \leq (1 + \epsilon) \dist_X(u,v)$,
		\item \label{prop:hd} the graph $G$ has highway dimension $hd_2 \in \mathcal{O}((\log(\alpha)/\epsilon)^d)$, and
		\item the graph $G$ has skeleton dimension $\kappa \in \mathcal{O}((\log(\alpha)/\epsilon)^d)$,
\end{alphaenumerate}
\end{lemma}

\begin{proof}
In~\cite{DBLP:journals/algorithmica/Feldmann19} it was shown, how to compute a Graph $H$ that satisfies properties \ref{prop:distortion} and \ref{prop:hd}. This was done by choosing so called \emph{hub sets} $Y_i \subseteq X$ for all $i = 0, 1, \dots, L = \lceil \log_2 \alpha \rceil$ such that in $H$ any shortest path in the range $(2^i, 2^{i+1}]$ contains some node from $Y_i$. Moreover, for any vertex $u \in X$ and any $i$ there is a hub $v \in Y_i$ satisfying $\dist_X(u,v) \leq \frac{\epsilon 2^{i-2}}{(1+\epsilon)^2 L}$ and for any two distinct hubs $u,v \in Y_i$ we have $\dist_X(u,v) > \frac{\epsilon 2^{i-3}}{(1+\epsilon)^2 L}$. The hub sets form a hierarchy, i.e.\ $Y_i \supseteq Y_j$ for all $i < j$. In the computed graph $H$, there is an edge between two vertices $u $ and $v$ of length $(1 + \epsilon(1-i/L)) \dist_X(u,v)$ if and only if $i = \max\{ j \mid \{u,v\} \subseteq Y_j\}$. We call such an edge also a \emph{level $i$ edge}. Moreover, $Y_0 = X$ is chosen and hence the graph $H$ has $\vert X \vert \choose 2$ edges, which implies a maximum degree of $(\vert X \vert - 1)$, which is a lower bound for the skeleton dimension.

However, we can observe that many of these edges are not necessary as they are not a shortest path. In particular, we can remove all level $i$ edges $\{u,v\}$ of $H$ satisfying $\dist_X(u,v) > 2^{i+1}$, which yields a graph $G$. The following claim shows that this does not affect the shortest path structure of the graph.

\begin{claim}
The graph $G$ fulfils properties \ref{prop:distortion} and \ref{prop:hd}.
\end{claim}

\begin{claimproof}
We show that the constructed graph $G$ has exactly the same shortest paths as $H$. This implies that $G$ fulfils properties \ref{prop:distortion} and \ref{prop:hd}.
Consider an edge $\{u,v\}$ that was removed from $H$. We claim that $\{u,v\}$ is longer than the shortest path from $u$ to $v$. As the edge was removed, we have $\dist_X(u,v) > 2^{i+1}$ where $i = \max\{j \mid \{u,v\} \subseteq Y_j\}$ is the level of $\{u,v\}$.
The length of $\{u,v\}$ in $H$ is $d_{uv} = (1 + \epsilon(1-i/L)) \dist_X(u,v)$. Any shortest path longer than $2^{i+1}$ contains some hub from $Y_{i+1}$. Hence, if $u,v \not \in Y_{i+1}$, the edge $\{u,v\}$ cannot be a shortest path in $H$ and we are done. Assume now that $u \in Y_{i+1}$. This implies $v \in Y_i \setminus Y_{i+1}$ as $\{u,v\}$ has level $i$. A property of the hub set $Y_{i+1}$ is that there is a hub $w \in Y_{i+1}$ satisfying $\dist_X(v,w) \leq \frac{\epsilon 2^{i-1}}{(1+\epsilon)^2 L}$.

As $u,w \in Y_{i+1}$, the edge $\{u,w\}$ has in $H$ length at most
\begin{align*}
d_{uw} & = \left(1 + \epsilon\left(1-\frac{i+1}{L}\right)\right) \dist_X(u,w) \\
       & \leq \left(1 + \epsilon\left(1-\frac{i+1}{L}\right)\right) (\dist_X(u,v) + \dist_X(v,w)) \\
	   & \leq \left(1 + \epsilon\left(1-\frac{i+1}{L}\right)\right) \dist_X(u,v) + 2 \cdot \dist_X(v,w).
\end{align*}
This means, that $H$ contains a path from $u$ to $w$ of length at most $d_{uw}$. Moreover, property~\ref{prop:distortion} implies that $H$ contains a $v$-$w$-path of length at most $d_{wv} = 2 \cdot \dist_X(v,w)$. It follows that by concatenating the shortest paths from $u$ to $w$ and from $w$ to $v$, we obtain a $u$-$v$-path whose length is upper bounded by
\begin{align*}
	d_{uw} + d_{wv} & \leq \left(1 + \epsilon\left(1-\frac{i+1}{L}\right)\right) \dist_X(u,v) + 2 \cdot \dist_X(v,w) + 2 \cdot \dist_X(v,w) \\
	& = \left(1 + \epsilon\left(1-\frac{i}{L}\right)\right) \dist_X(u,v) - \frac{\epsilon}{L} \cdot \dist_X(u,v) + 4 \cdot \dist_X(v,w)\\
	& \leq d_{uv} - \frac{\epsilon}{L} \cdot \dist_X(u,v) + 4 \cdot \dist_X(v,w) \\
	& < d_{uv} - \frac{\epsilon}{L} 2^{i+1} + 4 \cdot \frac{\epsilon 2^{i-1}}{(1+\epsilon)^2 L} \\
	& = d_{uv} - \frac{\epsilon}{L} \cdot \left(2^{i+1} - \frac{2^{i+1}}{(1+\epsilon)^2}\right) \\
	& < d_{uv}.
\end{align*}

Hence, the edge $\{u,v\}$ is longer than the shortest path from $u$ to $v$. 
\end{claimproof}

It follows that if $\{u,v\}$ is a long edge in $G$, then both $u$ and $v$ must be important hubs.

\begin{claim}\label{claim:length-level}
In $G$, for any edge $\{u,v\}$ of length more than $2^i$ we have $u,v \in Y_{i-1}$.
\end{claim}

\begin{claimproof}
Consider an edge $\{u,v\}$ of level $j \leq i-2$. As $\{u,v\}$ was not removed from $H$, we have $\dist_X(u,v) \leq 2^{j+1} \leq 2^{i-1}$. An upper bound of $2^i$ on the length of $\{u,v\}$ follows, as the length of $\{u,v\}$ was chosen as $(1 + \epsilon(1-j/L)) \cdot \dist_X(u,v) < 2 \cdot \dist_X(u,v) \leq 2^i$.
\end{claimproof}

It remains to bound the skeleton dimension of $G$. For some vertex $s$ and radius $r$ consider the set $\mathrm{Cut}_s^r$ in the skeleton of $s$ at radius $r$ and choose some vertex $v \in \mathrm{Cut}_s^r$.
Let $w$ be a furthest descendant of $v$ in the shortest path tree of $s$ and choose $i$ such that for the distance $r' = \dist_{G^*}(v,w)$ we have $2^i < r' \leq 2^{i+1}$. As $v$ is contained in the skeleton, it follows that for the distance $r = \dist_{G^*}(s,v)$ we have $r \leq 2r' \leq 2^{i+2}$.

Choose an edge $\{\parent{v}, \child{v}\}$ of $G$ satisfying $\dist_{G^*}(\parent{v},\child{v}) = \dist_{G^*}(\parent{v},v) + \dist_{G^*}(v,\child{v})$. In other words, $\parent{v}$ and $\child{v}$ are the parent and child node of $v$ when considering only nodes from the (discrete) graph $G$. We claim that the vertex $v$ has a descendant in the shortest path tree of $s$ which is contained in the hub set $Y_{i-2}$. To prove this, observe that $\dist_G(\parent{v}, w) \geq r' > 2^i$. This implies that (i) $\dist_G(\parent{v}, \child{v}) > 2^{i-1}$ or (ii) $\dist_G(\child{v}, w) > 2^{i-1}$.  Consider case (i). As the edge $\{\parent{v}, \child{v}\}$ is contained in $G$ and has length more than $2^{i-1}$, it follows from Claim~\ref{claim:length-level} that $\child{v} \in Y_{i-2}$. In case (ii), it follows from $\dist_G(\child{v}, w) > 2^{i-1}$ that the shortest path from $\child{v}$ to $w$ muss pass through a hub from $Y_{i-1} \subseteq Y_{i-2}$.

Denote the ball $B_r(s)$ in $G$ simply by $B_r$. It holds that $\dist_G(s,w) = r + r' < 2^{i+3}$. This means that every vertex $v \in \mathrm{Cut}_s^r$ has a descendant $v'$ in the shortest path tree which is contained in $Y_{i-2} \cap B_{2^{i+3}}$. As the skeleton of $s$ is a tree, for all distinct vertices $u,v \in \mathrm{Cut}_s^r$ we have $u' \neq v'$. Hence, we have $\vert \mathrm{Cut}_s^r \vert \leq \vert Y_{i-2} \cap B_{2^{i+3}} \vert$.

It was shown that if $(X, \dist_X)$ is a metric with doubling dimension $d$, for any subset $X' \subseteq X$ of aspect ratio $\beta$, the size of $X'$ is bounded by $2^{d \lceil \log_2 \beta \rceil} \leq (2 \beta)^{d}$ \cite{Gupta2003}. As the diameter of the ball $B_{2^{i+3}}$ is at most $2^{i+4}$ (which according to property~\ref{prop:distortion} also bounds the diameter of the ball w.r.t.\ the metric $(X, \dist_X)$) and any two distinct hubs $u, v \in Y_{i-2}$ have distance $\dist_X(u,v) > \frac{\epsilon 2^{i-5}}{(1+\epsilon)^2L}$, the aspect ratio of $Y_{i-2} \cap B_{2^{i+3}}$ w.r.t. $\dist_X$ is $\beta < 2^{i+4} / \frac{\epsilon 2^{i-5}}{(1+\epsilon)^2L} = 2^9 (1+\epsilon)^2 L / \epsilon$. It follows that $\vert Y_{i-2} \cap B_{2^{i+3}} \vert \leq (2 \cdot 2^9 (1+\epsilon)^2 L / \epsilon)^{d}$. As $\epsilon < 1$, the size of any $\mathrm{Cut}_s^r$ is bounded by $(2^{12} L / \epsilon)^{d}$ and we obtain that the skeleton dimension of $G$ is $\kappa \in \mathcal{O}((L/\epsilon)^d) = \mathcal{O}((\log(\alpha)/\epsilon)^d)$.
\end{proof}

Feldmann~\cite{DBLP:journals/algorithmica/Feldmann19} observed that due to a result of Feder and Greene~\cite{Feder1988}, it is NP-hard for any $\epsilon >0$ to compute a $(2-\epsilon)$-approximation for $k$-Center on graphs of doubling dimension $4$ and aspect ratio at most $n$. Lemma~\ref{lem:distortion} hence implies that it is NP-hard to compute a $(2-\epsilon)$-approximation if the skeleton dimension is in $\mathcal{O}(\log^2 n)$. It remains open whether this also holds for $\kappa \in o(\log^2 n)$ and in particular for constant skeleton dimension. 

It was also shown, that under the exponential time hypothesis (ETH) it is not possible to compute a $(2-\epsilon)$-approximation for $k$-Center on graphs of highway dimension $hd_2$ in time $2^{2^{o(\sqrt{hd_2})}} \cdot n^{\mathcal{O}(1)}$ \cite{DBLP:journals/algorithmica/Feldmann19}. Analogously, Lemma~\ref{lem:distortion} implies a bound of $2^{2^{o(\sqrt{\kappa})}} \cdot n^{\mathcal{O}(1)}$ for skeleton dimension $\kappa$. We summarize our findings in the following theorem.

\begin{theorem}
For any $\epsilon > 0$, it is NP-hard to compute a $(2-\epsilon)$-approximation for the $k$-Center problem on graphs of skeleton dimension $\kappa \in \mathcal{O}(\log^2 n)$. Assuming ETH there is no $2^{2^{o(\sqrt{\kappa})}} \cdot n^{\mathcal{O}(1)}$ time algorithm that computes a $(2-\epsilon)$-approximation.
\end{theorem}

\section{Conclusion and Future Work}
We showed that the skeleton dimension, the highway dimension (when defined as in \cite{abr10} or \cite{abr11}) and several other graph parameters are pairwise incomparable. Nevertheless, the skeleton dimension is upper bounded by the max leaf number and lower bounded through the maximum degree and the doubling dimension.

However, for the highway dimensions $hd_1$ and $hd_2$ there are still no tight upper or lower bounds. Using a grid graph and a complete graph, it can be shown that they are not even comparable to the minimum degree or the maximum clique size, which are lower bounds for a large number of graph parameters. Bauer et al.\ showed, that for any unweighted graph there are edge weights such that $hd_2 \geq (pw - 1) / (\log_{3/2} \vert V \vert + 2)$ where $pw$ is the pathwidth~\cite{Bauer2016}. It remains open whether this bound is tight.

It turned out that computing the highway dimension is NP-hard for all three different definitions used in the literature. Still, knowing the highway dimension of real-world networks will give further insight in the structure of transportation networks and hence it is worthwhile to study whether there are FPT algorithms to compute the highway dimension and to what extent it can be approximated.

We proved that on graphs of skeleton dimension $\mathcal{O}(\log^2 n)$ it is not possible to beat the well-known $2$-approximation algorithm by Hochbaum and Shmoys for \problem{$k$-Center}. Yet, the experimental results reported in \cite{blu18} indicate that the skeleton dimension of real-world networks might actually be a constant independent of the size of the network. This raises the question whether there is a $(2-\epsilon)$-approximation algorithm for graphs of constant skeleton dimension.

\bibliographystyle{plainurl}
\bibliography{bibliography}

\begin{thebibliography}{10}

\bibitem{abr16}
Ittai Abraham, Daniel Delling, Amos Fiat, Andrew~V. Goldberg, and Renato~F.
  Werneck.
\newblock Highway dimension and provably efficient shortest path algorithms.
\newblock {\em J. {ACM}}, 63(5):41:1--41:26, 2016.
\newblock \href {http://dx.doi.org/10.1145/2985473}
  {\path{doi:10.1145/2985473}}.

\bibitem{abr11}
Ittai Abraham, Daniel Delling, Amos Fiat, Andrew~V. Goldberg, and Renato
  Fonseca~F. Werneck.
\newblock Vc-dimension and shortest path algorithms.
\newblock In {\em Proceedings of the 38th International Colloquium on Automata,
  Languages and Programming ({ICALP})}, pages 690--699, 2011.
\newblock \href {http://dx.doi.org/10.1007/978-3-642-22006-7\_58}
  {\path{doi:10.1007/978-3-642-22006-7\_58}}.

\bibitem{abr10}
Ittai Abraham, Amos Fiat, Andrew~V. Goldberg, and Renato Fonseca~F. Werneck.
\newblock Highway dimension, shortest paths, and provably efficient algorithms.
\newblock In {\em Proceedings of the 21st Annual {ACM-SIAM} Symposium on
  Discrete Algorithms, ({SODA})}, pages 782--793, 2010.
\newblock \href {http://dx.doi.org/10.1137/1.9781611973075.64}
  {\path{doi:10.1137/1.9781611973075.64}}.

\bibitem{DBLP:journals/bit/Arnborg85}
Stefan Arnborg.
\newblock Efficient algorithms for combinatorial problems with bounded
  decomposability - {A} survey.
\newblock {\em {BIT}}, 25(1):2--23, 1985.

\bibitem{Arora1998}
Sanjeev Arora.
\newblock Polynomial time approximation schemes for euclidean traveling
  salesman and other geometric problems.
\newblock {\em J. {ACM}}, 45(5):753--782, 1998.
\newblock \href {http://dx.doi.org/10.1145/290179.290180}
  {\path{doi:10.1145/290179.290180}}.

\bibitem{Bartal2016}
Yair Bartal, Lee{-}Ad Gottlieb, and Robert Krauthgamer.
\newblock The traveling salesman problem: Low-dimensionality implies a
  polynomial time approximation scheme.
\newblock {\em {SIAM} J. Comput.}, 45(4):1563--1581, 2016.
\newblock \href {http://dx.doi.org/10.1137/130913328}
  {\path{doi:10.1137/130913328}}.

\bibitem{bas06}
H.~Bast, Stefan Funke, and Domagoj Matijevic.
\newblock Ultrafast shortest-path queries via transit nodes.
\newblock In {\em The Shortest Path Problem, Proceedings of a {DIMACS}
  Workshop}, volume~74 of {\em {DIMACS} Series in Discrete Mathematics and
  Theoretical Computer Science}, pages 175--192. {DIMACS/AMS}, 2006.

\bibitem{bas07}
H.~Bast, Stefan Funke, Domagoj Matijevic, Peter Sanders, and Dominik Schultes.
\newblock In transit to constant time shortest-path queries in road networks.
\newblock In {\em Proceedings of the 9th Workshop on Algorithm Engineering and
  Experiments ({ALENEX})}. {SIAM}, 2007.
\newblock \href {http://dx.doi.org/10.1137/1.9781611972870.5}
  {\path{doi:10.1137/1.9781611972870.5}}.

\bibitem{Bauer2016}
Reinhard Bauer, Tobias Columbus, Ignaz Rutter, and Dorothea Wagner.
\newblock Search-space size in contraction hierarchies.
\newblock {\em Theor. Comput. Sci.}, 645:112--127, 2016.
\newblock \href {http://dx.doi.org/10.1016/j.tcs.2016.07.003}
  {\path{doi:10.1016/j.tcs.2016.07.003}}.

\bibitem{Becker2018}
Amariah Becker, Philip~N. Klein, and David Saulpic.
\newblock Polynomial-time approximation schemes for k-center, k-median, and
  capacitated vehicle routing in bounded highway dimension.
\newblock In Yossi Azar, Hannah Bast, and Grzegorz Herman, editors, {\em
  Proceedings of the 26th Annual European Symposium on Algorithms ({ESA})},
  volume 112 of {\em LIPIcs}, pages 8:1--8:15. Schloss Dagstuhl -
  Leibniz-Zentrum fuer Informatik, 2018.
\newblock \href {http://dx.doi.org/10.4230/LIPIcs.ESA.2018.8}
  {\path{doi:10.4230/LIPIcs.ESA.2018.8}}.

\bibitem{blu18}
Johannes Blum and Sabine Storandt.
\newblock Computation and growth of road network dimensions.
\newblock In Lusheng Wang and Daming Zhu, editors, {\em Proceedings of the 24th
  International Computing and Combinatorics Conference ({COCOON})}, volume
  10976 of {\em Lecture Notes in Computer Science}, pages 230--241. Springer,
  2018.
\newblock \href {http://dx.doi.org/10.1007/978-3-319-94776-1\_20}
  {\path{doi:10.1007/978-3-319-94776-1\_20}}.

\bibitem{DBLP:conf/icalp/Bodlaender88}
Hans~L. Bodlaender.
\newblock Dynamic programming on graphs with bounded treewidth.
\newblock In Timo Lepist{\"{o}} and Arto Salomaa, editors, {\em Proceedings of
  the 15th International Colloquium on Automata, Languages and Programming
  ({ICALP})}, volume 317 of {\em Lecture Notes in Computer Science}, pages
  105--118. Springer, 1988.
\newblock \href {http://dx.doi.org/10.1007/3-540-19488-6\_110}
  {\path{doi:10.1007/3-540-19488-6\_110}}.

\bibitem{Bodlaender1998}
Hans~L. Bodlaender.
\newblock A partial \emph{k}-arboretum of graphs with bounded treewidth.
\newblock {\em Theor. Comput. Sci.}, 209(1-2):1--45, 1998.
\newblock \href {http://dx.doi.org/10.1016/S0304-3975(97)00228-4}
  {\path{doi:10.1016/S0304-3975(97)00228-4}}.

\bibitem{cat96}
Kevin Cattell, Michael~J. Dinneen, and Michael~R. Fellows.
\newblock A simple linear-time algorithm for finding path-decompositions of
  small width.
\newblock {\em Inf. Process. Lett.}, 57(4):197--203, 1996.
\newblock \href {http://dx.doi.org/10.1016/0020-0190(95)00190-5}
  {\path{doi:10.1016/0020-0190(95)00190-5}}.

\bibitem{DeLaVina2009}
Ermelinda DeLaViña and Bill Waller.
\newblock A note on a conjecture of hansen et al.
\newblock Unpublished manuscript, 2009.
\newblock URL:
  \url{http://cms.dt.uh.edu/faculty/delavinae/research/DelavinaWaller2009.pdf}.

\bibitem{Feder1988}
Tom{\'{a}}s Feder and Daniel~H. Greene.
\newblock Optimal algorithms for approximate clustering.
\newblock In Janos Simon, editor, {\em Proceedings of the 20th Annual {ACM}
  Symposium on Theory of Computing ({STOC})}, pages 434--444. {ACM}, 1988.
\newblock \href {http://dx.doi.org/10.1145/62212.62255}
  {\path{doi:10.1145/62212.62255}}.

\bibitem{DBLP:journals/algorithmica/Feldmann19}
Andreas~Emil Feldmann.
\newblock Fixed-parameter approximations for k-center problems in low highway
  dimension graphs.
\newblock {\em Algorithmica}, 81(3):1031--1052, 2019.
\newblock \href {http://dx.doi.org/10.1007/s00453-018-0455-0}
  {\path{doi:10.1007/s00453-018-0455-0}}.

\bibitem{Fel15b}
Andreas~Emil Feldmann, Wai~Shing Fung, Jochen K{\"{o}}nemann, and Ian Post.
\newblock A {(1} + {\(\epsilon\)})-embedding of low highway dimension graphs
  into bounded treewidth graphs.
\newblock {\em CoRR}, abs/1502.04588, 2015.
\newblock \href {http://arxiv.org/abs/1502.04588} {\path{arXiv:1502.04588}}.

\bibitem{DBLP:journals/siamcomp/FeldmannFKP18}
Andreas~Emil Feldmann, Wai~Shing Fung, Jochen K{\"{o}}nemann, and Ian Post.
\newblock A (1+{\(\epsilon\)})-embedding of low highway dimension graphs into
  bounded treewidth graphs.
\newblock {\em {SIAM} J. Comput.}, 47(4):1667--1704, 2018.
\newblock \href {http://dx.doi.org/10.1137/16M1067196}
  {\path{doi:10.1137/16M1067196}}.

\bibitem{Fel18}
Andreas~Emil Feldmann and D{\'{a}}niel Marx.
\newblock The parameterized hardness of the k-center problem in transportation
  networks.
\newblock In David Eppstein, editor, {\em Proceedings of the 16th Scandinavian
  Symposium and Workshops on Algorithm Theory ({SWAT})}, volume 101 of {\em
  LIPIcs}, pages 19:1--19:13. Schloss Dagstuhl - Leibniz-Zentrum fuer
  Informatik, 2018.
\newblock \href {http://dx.doi.org/10.4230/LIPIcs.SWAT.2018.19}
  {\path{doi:10.4230/LIPIcs.SWAT.2018.19}}.

\bibitem{Gottlieb2013}
Lee{-}Ad Gottlieb and Robert Krauthgamer.
\newblock Proximity algorithms for nearly doubling spaces.
\newblock {\em {SIAM} J. Discrete Math.}, 27(4):1759--1769, 2013.
\newblock \href {http://dx.doi.org/10.1137/120874242}
  {\path{doi:10.1137/120874242}}.

\bibitem{Gupta2003}
Anupam Gupta, Robert Krauthgamer, and James~R. Lee.
\newblock Bounded geometries, fractals, and low-distortion embeddings.
\newblock In {\em Proceedings of the 44th Symposium on Foundations of Computer
  Science {(FOCS})}, pages 534--543. {IEEE} Computer Society, 2003.
\newblock \href {http://dx.doi.org/10.1109/SFCS.2003.1238226}
  {\path{doi:10.1109/SFCS.2003.1238226}}.

\bibitem{Hochbaum1986}
Dorit~S. Hochbaum and David~B. Shmoys.
\newblock A unified approach to approximation algorithms for bottleneck
  problems.
\newblock {\em J. {ACM}}, 33(3):533--550, 1986.
\newblock \href {http://dx.doi.org/10.1145/5925.5933}
  {\path{doi:10.1145/5925.5933}}.

\bibitem{kap96}
Haim Kaplan and Ron Shamir.
\newblock Pathwidth, bandwidth, and completion problems to proper interval
  graphs with small cliques.
\newblock {\em {SIAM} J. Comput.}, 25(3):540--561, 1996.
\newblock \href {http://dx.doi.org/10.1137/S0097539793258143}
  {\path{doi:10.1137/S0097539793258143}}.

\bibitem{kos16}
Adrian Kosowski and Laurent Viennot.
\newblock Beyond highway dimension: Small distance labels using tree skeletons.
\newblock In {\em Proceedings of the 28th Annual ACM-SIAM Symposium on Discrete
  Algorithms ({SODA})}, pages 1462--1478. SIAM, 2017.

\bibitem{Plesnik1980}
Jan Plesn{\'\i}k.
\newblock On the computational complexity of centers locating in a graph.
\newblock {\em Aplikace matematiky}, 25(6):445--452, 1980.
\newblock URL:
  \url{https://dml.cz/bitstream/handle/10338.dmlcz/103883/AplMat_25-1980-6_8.pdf}.

\bibitem{Sorge2019}
Manuel Sorge, Matthias Weller, Florent Foucaud, Ondřej Suchý, Pascal Ochem,
  Martin Vatshelle, and Gerhard~J. Woeginger.
\newblock The graph parameter hierarchy.
\newblock Unpublished manuscript, 2019.
\newblock URL: \url{https://manyu.pro/assets/parameter-hierarchy.pdf}.

\bibitem{Talwar2004}
Kunal Talwar.
\newblock Bypassing the embedding: algorithms for low dimensional metrics.
\newblock In L{\'{a}}szl{\'{o}} Babai, editor, {\em Proceedings of the 36th
  Annual {ACM} Symposium on Theory of Computing ({SODA})}, pages 281--290.
  {ACM}, 2004.
\newblock \href {http://dx.doi.org/10.1145/1007352.1007399}
  {\path{doi:10.1145/1007352.1007399}}.

\bibitem{Vazirani2001}
Vijay~V. Vazirani.
\newblock {\em Approximation algorithms}.
\newblock Springer, 2001.
\newblock URL:
  \url{http://www.springer.com/computer/theoretical+computer+science/book/978-3-540-65367-7}.

\end{thebibliography}
\end{document}